\documentclass[11pt]{article}
\usepackage{amsmath,amssymb,amsthm,algorithm,algorithmic,prettyref,boxedminipage}
\usepackage{fullpage,pstricks,pst-node,url}
\usepackage{xcolor}
\definecolor{foocite}{rgb}{0,0.75,0}
\definecolor{foolink}{rgb}{0,0,1}
\definecolor{foourl}{rgb}{1,0,0}
\usepackage[colorlinks=true,citecolor=foocite,linkcolor=foolink,urlcolor=foourl]{hyperref}
\newcommand{\comment}[1]{}
\newcommand{\ignore}[1]{}
\newtheorem{theorem}{Theorem}
\newrefformat{theorem}{Theorem~\ref{#1}}
\newrefformat{eq}{Equation~\eqref{#1}}
\newrefformat{chap}{Chapter~\ref{#1}}
\newrefformat{fig}{Figure~\ref{#1}}
\def\xthm[#1][#2][#3]{\newtheorem{#2}[theorem]{#3} \newrefformat{#2}{#3 \ref{#11}}}
\xthm[#][defn][Definition]
\xthm[#][claim][Claim]
\xthm[#][conj][Conjecture]
\xthm[#][prop][Proposition]
\xthm[#][lemma][Lemma]
\xthm[#][eg][Example]
\xthm[#][fact][Fact]
\xthm[#][cor][Corollary]
\xthm[#][app][Appendix]
\xthm[#][alg][Algorithm]
\xthm[#][step][Step]
\xthm[#][remark][Remark]

\newcommand{\OPT}{\ensuremath{\mathrm{OPT}}}
\newcommand{\R}{\ensuremath{\mathbb{R}}}
\newcommand{\Z}{\ensuremath{\mathbb{Z}}}
\newcommand{\vz}{\ensuremath{\mathbf{0}}}
\newcommand{\vo}{\ensuremath{\mathbf{1}}}

\newcommand{\NP}{\ensuremath{\mathsf{NP}}}
\newcommand{\PP}{\ensuremath{\mathsf{P}}}

\newcommand{\K}{\ensuremath{\mathcal{K}}}
\renewcommand{\L}{\ensuremath{\mathcal{L}}}
\renewcommand{\P}{\ensuremath{\mathcal{P}}}

\begin{document}
\title{An LP with Integrality Gap $1+\epsilon$ for Multidimensional Knapsack}
\author{David Pritchard\footnote{\'Ecole Polytechnique F\'ed\'erale de Lausanne and partially supported by an NSERC post-doctoral fellowship.}}
\date{\today}

\maketitle
\begin{abstract}
In this note we study packing or covering integer programs with at most $k$ constraints, which are also known as \emph{$k$-dimensional knapsack problems}. For integer $k > 0$ and real $\epsilon > 0 $, we observe there is a polynomial-sized LP for the $k$-dimensional knapsack problem with integrality gap at most $1+\epsilon$. The variables may be unbounded or have arbitrary upper bounds. In the (classical) packing case, we can also remove the dependence of the LP on the cost-function, yielding a polyhedral approximation of the integer hull. This generalizes a recent result of Bienstock~\cite{Bienstock08} on the classical knapsack problem.
\end{abstract}

\section{Introduction}
The classical \emph{knapsack problem} is the following: given a collection of items each with a value and a weight, and given a weight limit, find a subset of items whose total weight is at most the weight limit, and whose value is maximized. If $n$ denotes the number of items, this can be formulated as the integer program $\{\max \sum_{i=1}^n x_iv_i \mid x \in \{0,1\}^n, \sum_{i=1}^n x_i w_i \le \ell\}$ where $n$ denotes the number of items, $v_i$ denotes the value of item $i$, $w_i$ denotes the weight of item $i$, and $\ell$ denotes the weight limit.

In the more general \emph{$k$-dimensional knapsack} (or $k$-constrained knapsack) problem, there are $k$ different kinds of ``weight" and a limit for each kind. An example for $k=3$ would be a robber who is separately constrained by the total mass, volume, and noisiness of the items he is choosing to steal. An orthogonal generalization is that the robber could take multiple copies of each item $i$, up to some prescribed limit of $d_i$ available copies. We therefore model the $k$-dimensional knapsack problem as
\begin{equation}
\{\max cx \mid x \in \Z^n, 0 \le x \le d, Ax \leq b\}\label{eq:ilpknap}\end{equation}
where $A$ is a $k$-by-$n$ matrix, $b$ is a vector of length $k$, and $d$ is a vector of length $n$, all non-negative and integral. Two special cases are common: if $d = \vo$ we call it the \emph{$0\textrm{-}1$ knapsack problem}; if $d = +\infty$, we call it the \emph{unbounded knapsack problem}.

Another natural generalization is the \emph{$k$-dimensional knapsack-cover problem},
$$\{\min cx \mid x \in \Z^n, 0 \le x \le d, Ax \geq b\}$$
which has analogous unbounded and 0-1 special cases. We sometimes call this version the \emph{covering version} and likewise \eqref{eq:ilpknap} is the \emph{packing version}.

On the positive side, for any fixed $k$, all above variants admit a simple pseudo-polynomial-time dynamic programming solution.
Chandra et al.~\cite{CHW76} gave the first PTAS (polynomial-time approximation scheme) for $k$-dimensional knapsack in 1976, and later an LP-based scheme was given by Frieze and Clarke~\cite{FC84}. See the book by Kellerer et al.~\cite[\S 9.4.2]{KPP04} for a more comprehensive literature review.
The case $k=1$ also admits a fully polynomial-time approximation scheme (FPTAS), but for $k \ge 2$ there is no FPTAS unless \PP=\NP. This was originally shown for 0-1 $k$-dimensional knapsack by Gens \& Levner~\cite{GL79} and Korte \& Schrader~\cite{KS80} (see also \cite{KPP04}) and subsequently for arbitrary $d$ by Magazine \& Chern~\cite{MC84}.

Our main result is the following:
\begin{theorem}\label{theorem:maint}
Let $k$ and $\epsilon$ be fixed. Given a $k$-dimensional knapsack (resp.~knapsack-cover) instance $\K$, there is a polynomial-sized extended LP relaxation $\L$ of $\P$ with $\OPT(\P) \ge (1-\epsilon)\OPT(\L)$ (resp.~with $\OPT(\P) \le (1+\epsilon)\OPT(\L)$).
\end{theorem}
Here ``polynomial-sized extended LP relaxation" means the following. First, $\P$ has $n$ variables. Then $\L$ must have those $n$ variables plus a polynomial number of other ones. The projection $\L'$ of $\L$ onto the first $n$ variables must contain the same integral solutions as $\P$. Finally, $\L$ and $\P$ must have the same objective function, i.e.~the objective function should ignore the extended variables.

In the proof, we will see that the LP can be constructed in polynomial time, and that a near-optimal integral solution can be obtained from an optimal extreme point fractional solution just by rounding down (resp.~up). The number of variables in the LP is $n^{O(k/\epsilon)}$ and the number of constraints is $kn^{O(k/\epsilon)}$. The \emph{integrality gap} of an IP is the worst-case ratio between the fractional and integral optimum and therefore \prettyref{theorem:maint} can be equivalent stated as saying that $\P$ has integrality gap at most $1+\epsilon$.

Our result and the techniques we use are a generalization of a recent result of Bienstock~\cite{Bienstock08}, which dealt with the packing version for $k=1$. The key observation we contribute is that his ``filtering" approach was also traditionally used to get a PTAS for multi-dimensional knapsack; in \emph{filtering} we exhaustively guess the $\gamma$ max-cost items in the knapsack for some constant $\gamma$.

The construction of $\L$ in \prettyref{theorem:maint} turns out to depend on the cost function $c$. A more interesting and challenging problem is to find an $\L$ which is independent of the cost-function, since this gives a \emph{polyhedral approximation} $\L'$ of $\P$ e.g.~in the packing case, it implies $\L' \supset \P \supset (1-\epsilon)\L'$.
Bienstock's result~\cite{Bienstock08} actually gives an LP which does not depend on the item cost/profits $c$. We will show (in \prettyref{sec:indep}) that in the packing case, our approach can be similarly revised:
\begin{theorem}\label{theorem:maintwo}
Let $k$ and $\epsilon$ be fixed. Given a $k$-dimensional knapsack instance $\K$, there is a polynomial-sized extended LP relaxation $\L$ of $\P$ with $\OPT(\P) \ge (1-\epsilon)\OPT(\L)$, such that $\L$ does not depend on $c$. 
\end{theorem}
This comes as the cost of an increase in size to $kn^{O(k^2/\epsilon)}$. For the covering case performing the same (a polynomial-sized extended LP relaxation independent of $c$ with integrality gap $\le 1+\epsilon$) is an interesting open problem; we elaborate at the end.

\subsection{Related Work}
Knapsack (whether packing or covering) has an FPTAS by dynamic programming, and it is well-known that dynamic programs of such a form can be solved as a shortest-path problem, which has an LP formulation. Nonetheless, there is no evident way to combine these steps to get an LP for knapsack with integrality gap $1+\epsilon$. The problem (say, for packing, which is simpler) is that last step in the FPTAS is not merely to return the last entry of the DP table, but rather it finds the maximum scaled profit such that the minimum volume to obtain it fits inside the knapsack (and then recovers the actual solution). The naive fix is adding this volume constraint to the LP but it makes the LP non-integral and then it is not clear how to proceed.

Bienstock \& McClosky~\cite{BM08} extend the work of Bienstock~\cite{Bienstock08} to covering problems and other settings, and also give an LP of size $n^2(1/\epsilon)^{\frac{1}{\epsilon}\log \frac{1}{\epsilon}}$ with integrality gap $1+\epsilon$ for 1-dimensional, 0-1 covering knapsack.\footnote{They use a disjunctive program; in essence, the LP guesses the most costly item in the knapsack, then for $i = 1, \dotsc, O(\frac{1}{\epsilon}\log\frac{1}{\epsilon})$ it guesses the number of items whose costs are $(1+\frac{1}{\epsilon})^{-(i, i+1]}$ times that cost, with all guesses $> \frac{1}{\epsilon}$ deemed equivalent. In particular the LP depends on the cost function. We remark that the method does not readily extend to $k$-dimensional knapsack. } There is some current work \cite{CSh08} on obtaining primal-dual algorithms (that is, not needing the ellipsoid method or interior-point subroutines) for knapsack-type covering problems with good approximation ratio and \cite{BM08} reports that the methods of \cite{CSh08} extend to a combinatorial LP-based approximation scheme for 1-dimensional covering knapsack.

Answering an open question of Bienstock~\cite{BM08} about the efficacy of automatic relaxations for the knapsack problem, Karlin et al.~\cite{KMN10} recently found that the ``Laserre hierarchy" of semidefinite programming relaxations, when applied to the 1-dimensional 0-1 packing knapsack problem, gives an SDP with integrality gap $1+\epsilon$ after $O(1/\epsilon^2)$ rounds.

Knapsack problems have a couple of interesting basic properties. The first contrasts with Lenstra's result~\cite{Lenstra1983} that for any fixed $k$, integer programs with $k$ constraints can be solved in polynomial time; in comparison, if we have nonnegativity constraints for every variable plus \emph{one other constraint}, we get the unbounded (1-dimensional) knapsack problem, which is \NP-hard~\cite{Lueker75}. Second, recall that for any optimization problem whose objective is integral, and whose optimal value is polynomial in the input size, any FPTAS can be used to get a pseudopolynomial-time algorithm. In contrast, 0-1 2-dimensional knapsack shows the converse is false: it has a pseudopolynomial-time algorithm, but getting an FPTAS is \NP-hard even when each profit $c_i$ is 1, e.g.~see~\cite[Thm.~9.4.1]{KPP04}.

\comment{We give two other recent developments in this field. There is a line of work in \emph{counting} the number of feasible solutions to a given $k$-dimensional knapsack problem (in which case there is no objective function $c$) and Dyer~\cite{Dyer03} recently gave a simple dynamic programming-based FPRAS (fully-polynomial time randomized approximation scheme) to count the number of feasible solutions for $k$-dimensional bounded packing knapsack. Separate from this,}
There is a line of work on maximizing constrained submodular functions. For non-monotone submodular maximization subject to $k$ linear packing constraints, the state of the art is by Lee et al.~\cite{LMNS09} who give a $(5+\epsilon)$-approximation algorithm. For monotone submodular maximization the state of the art is by Chekuri \& Vondr\'{a}k~\cite{CV09} who give a $(e/(e-1)+\epsilon)$-approximation subject to $k$ knapsack constraints and a matroid constraint. We note it is \NP-hard to obtain any factor better than $e/(e-1)$ for monotone submodular maximization over a matroid~\cite{F98}, so in this setting knapsack constraints only affect the best ratio by $\epsilon$, just like in our setting of LP-relative approximation.

\subsection{Overview}
First, we review rounding and filtering. Rounding is a standard approach to turn an optimal fractional solution into a nearly-optimal integral one, and here we lose up to $k$ times the maximum per-item profit. Filtering works well with rounding because it reduces the maximum per-item profit; the power of these ideas is already enough to get an LP-based approximation scheme~\cite{FC84}, but it uses a separate LP for each ``guess" made in filtering. Therefore, like Bienstock~\cite{BM08}, we use disjunctive programming~\cite{Ba79} to combine all the separate LPs into a single one. The approach has some similarity to the knapsack-cover inequalities of Carr et al.~\cite{CFLP00}.

\section{Rounding and Filtering}\label{sec:knap-ptas}
We now explain the approach.
A knapsack instance \eqref{eq:ilpknap} is determined by the parameters $(A, b, c, d)$. The na\"ive LP relaxation of the knapsack problem is
\begin{equation}\{\max cx \mid x \in \R^n, 0 \le x \le d, Ax \leq b\}.\label{eq:kdimpacklp}\tag*{$\K(A, b, c, d)$}\end{equation}
In the following, \emph{fractional} means non-integral. The following lemma is standard.

\begin{lemma}
Let $x^*$ be an extreme point solution to the linear program \eqref{eq:kdimpacklp}. Then $x^*$ is fractional in at most $k$ coordinates. \label{lemma:knapstruct}
\end{lemma}
\begin{proof}
It follows from elementary LP theory that $x^* \in \R^n$ satisfies $n$ (linearly independent) constraints with equality. There are $k$ constraints of the form $A_jx \le b_j$; all other constraints are of the form $x_i \ge 0$ or $x_i \le d_i$, so at least $n-k$ of them hold with equality. Clearly $x_i \ge 0$ and $x_i \le d_i$ cannot both hold with equality for the same $i$, so $x^*_i \in \{0, d_i\}$ for at least $n-k$ distinct $i$, which gives the result.
\end{proof}

Therefore, we obtain the following primitive guarantee on a rounding strategy. Let $\lfloor \cdot \rfloor$ applied to a vector mean component-wise floor and let $c_{\max} := \max_i c_i$.
\begin{cor}\label{cor:knapround}
Let $x^*$ be an extreme point solution to the linear program \eqref{eq:kdimpacklp}. Then $c \lfloor x^* \rfloor \ge cx^* - kc_{\max}$.
\end{cor}
Now the idea is to take $x^*$ to be an optimal fractional solution, and use filtering (exhaustive guessing) to turn the additive guarantee into a multiplicative factor of $1+\epsilon$. Let $\gamma$ denote a parameter, which represents the size of a multi-set we will guess. For a non-negative vector $z$ let the notation $\lVert z \rVert_1$ mean $\sum_i z_i$. A \emph{guess} is an integral vector $g$ with $0 \le g \le d, Ag \le b$ and $\lVert g \rVert_1 \le \gamma$. It is easy to see the number of possible guesses is bounded by $(n+1)^\gamma$, and that for any constant $\gamma$ we can iterate through all guesses in polynomial time.

From now on we assume without loss of generality (by reordering items if necessary) that $c_1 \le c_2 \le \dotsb \le c_n$.
For a guess $g$ with $\lVert g \rVert_1 = \gamma$ we now define the \emph{residual knapsack problem} for $g$. The residual problem models how to optimally select the remaining objects \emph{under the restriction} that the $\gamma$ most profitable\footnote{To simplify the description, even if $c_{i+1} = c_i$ we think of item $i+1$ as more profitable than item $i$.} items chosen (counting multiplicity) are $g$. Let $\mu(g)$ denote $\min \{i \mid g_i > 0\}$.
Define $d^g$ to be the first $\mu(g)$ coordinates of $d-g$ followed by $n-\mu(g)$ zeroes, and $b^g = b - Ag$. The \emph{residual knapsack problem} for $g$ is $(A, b^g, c, d^g)$.
The residual problem for $g$ does not permit taking items with index more than $\mu(g)$ and so its $c_{\max}$ value may be thought of as $c_{\mu(g)}$ or less, which is at most $c \cdot g / \lVert g \rVert_1 = c \cdot g / \gamma$.

If a guess $g$ has $\lVert g \rVert_1 < \gamma$, define $b^g$ and $d^g$ to be all-zero. Then \prettyref{cor:knapround} gives the following.
\begin{cor}\label{cor:knapoptround}
Let $x_{\OPT}$ be an optimal integral knapsack solution for $(A, b, c, d)$. Let
$g$ be the $\gamma$ most profitable items in $x_{\OPT}$ (or all, if there are less than $\gamma$). Let $x^*$ be an optimal extreme point solution to $\K(A, b^g, c, d^g)$. Then $g + \lfloor x^* \rfloor$ is a feasible knapsack solution for $(A, b, c, d)$ with value at least $1-k/\gamma$ times optimal.
\end{cor}
\begin{proof}
We use $\OPT$ to denote $c \cdot x_{\OPT}$.
Note that $x_{\OPT}-g$ is feasible for the residual problem for $g$. Therefore $c \cdot x^* \ge \OPT - c \cdot g$. Moreover $c_{\max}$ in the residual problem for $g$ is not more than $c \cdot g / \gamma \le \frac{\OPT}{\gamma}$, so \prettyref{cor:knapround} shows that $$c \cdot \lfloor x^* \rfloor \ge c \cdot x^* - k \frac{\OPT}{\gamma} \ge \OPT - c \cdot g - k \frac{\OPT}{\gamma}$$
and consequently $\lfloor x^* \rfloor + g$ is a solution with value at least $\OPT(1-\frac{k}{\gamma})$, as needed.
\end{proof}
By taking $\gamma = k/\epsilon$ and solving $\K(A, b^g, c, d^g)$ for all possible $g$ we get the previously known PTAS for $k$-dimensional knapsack; we now refine the approach to get a single LP.

\section{Disjunctive Programming}
We now review some disjunctive programming tools~\cite{Ba79}. The only result we need is that it is possible to write a compact LP for the convex hull of the union of several polytopes, provided that we we have compact LPs for each one.

Suppose we have polyhedra $P^1 = \{x \in \R^n \mid A^1x \le b^1\}$ and $P^2 = \{x \in \R^n \mid A^2x \le b^2\}$. Both of these sets are convex and it is therefore easy to see that the convex hull of their union is the set
$$\textrm{conv.hull}(P^1 \cup P^2) = \{x \in \R^n \mid x = \lambda x^1 + (1-\lambda) x^2, 0 \le \lambda \le 1, A^1x^1 \le b^1, A^1x^2 \le b^2\}.$$
However, this is not a \emph{linear} program, e.g.\ since we multiply the variable $\lambda$ by the variables $x^1$. Nonetheless, it is not hard to see that the following is a linear formulation of the same set:
$$\textrm{conv.hull}(P^1 \cup P^2) = \{x \in \R^n \mid x = x^1 + x^2, 0 \le \lambda \le 1, A^1x^1 \le \lambda b^1, A^1x^2 \le (1-\lambda)b^2\}.$$
A similar construction gives the convex hull of the union of any number of polyhedra; we now apply this to the knapsack setting.

The LP $\K(A, b^g, c, d^g)$ was constructed to mean the left-over problem after making a guess $g$ of the $\gamma$ most profitable items; we similarly shift this LP to get $\{y = x + g \mid x \in \R^n, 0 \le x \le d^g, Ax \leq b^g\}$ which is the same set, after the guessed part is added back in.

Let $\mathcal G$ denote the set of all possible guesses $g$. Then the convex hull of the union of the shifted polyhedra is given by the feasible region of the following polyhedron:
\begin{equation}\Bigl\{y \mid  y = \sum_{g \in \mathcal G} y^g; \sum_{g \in \mathcal G} \lambda^g = 1; \lambda \ge \vz; \forall g:  y^g = x^g + \lambda^g g, \vz \le x^g \le \lambda^gd^g, Ay^g \le \lambda^gb^g \Bigr\}.\label{eq:convun}\tag{\L}\end{equation}
We attach objective $\max c \cdot y$ to \eqref{eq:convun} to make it into an LP, and use it to prove \prettyref{theorem:maint}.

%Proof BLAH

%\begin{theorem}
%If $\gamma = k/\epsilon$, there is a polynomial-time \eqref{eq:convun}-relative $1/(1-\epsilon)$-approximation algorithm for the (0-1, unbounded, or bounded) $k$-dimensional knapsack problem.
%\end{theorem}
\begin{proof}[Proof of \prettyref{theorem:maint}, packing version]
Let $y$ be an optimal extreme point solution for \eqref{eq:convun}. It is straightforward to argue that any extreme point solution has $\lambda^{g^*} = 1$ for some particular $g^*$, and $\lambda^{g} = 0$ for all other $g$. Hence $y = x^{g^*} + g^*$ where $x^{g^*}$ is an optimal extreme point solution to $\K(A, b^{g^*}, c, d^{g^*})$. We now show that $\lfloor y \rfloor$ is a $(1-\epsilon)$-approximately optimal solution, re-using the previous arguments.

If $\lVert g^* \rVert_1 < \gamma$, then $x^{g^*} = 0$ so $y$ is integral, hence $y$ is an optimal knapsack solution. Otherwise, if $\lVert g^* \rVert_1 = \gamma$, then \prettyref{cor:knapround} shows that
$$c \cdot \lfloor y \rfloor = c\cdot \lfloor x^{g^*} \rfloor + c\cdot g^* \ge c \cdot x^{g^*} - k\frac{c \cdot g^*}{\gamma}  + c \cdot g^* = c \cdot y - k\frac{c \cdot g^*}{\gamma}  \ge (1-\epsilon) c \cdot y,$$
which completes the proof.
\end{proof}

The corresponding result for the covering version is very similar. One difference is that we round up instead of down. The other is that some guesses become inadmissible. Let $g$ be an integral vector with $0 \le g \le d, \lVert g \rVert_1 \le \gamma$; we define $\mu(g), d^g$ as before and call $g$ a \emph{guess} only if $A(g + d^g) \ge b$, in which case we set $b^g$ to be the component-wise maximum of $\vz$ and $b-Ag$.

\section{Removing the Dependence on $c$ for Packing Problems}\label{sec:indep}
In the LPs described above, for each guess $g$, we treated that guess as the set of most \emph{profitable} items. In particular, $b^g$ and $d^g$ are defined in a way that depends on $c$. We now show in the packing case, how to write a somewhat larger LP, still with integrality gap $1+\epsilon$, which is defined independently of $c$. This exactly follows the approach of Bienstock~\cite{Bienstock08}; what we will do is guess the \emph{biggest} items for each constraint, rather than the most profitable items. The technique does not seem to have an easy analogue for covering problems.

In detail, previously, we guessed the multiset $g$ of $\gamma$ most profitable items in the solution. Instead, let us guess a $k$-tuple $(g^1, g^2, \dotsc, g^k)$ where for each $k$, $g^i$ is the set of $\gamma$ items in the solution which have largest coefficients with respect to the $i$th constraint (breaking ties in each constraint in any consistent way). What we need is that any extreme feasible solution with at most $k$ fractional values can be rounded to an integral feasible solution at a relative cost factor of at most $\epsilon$. Let the original extreme point LP solution be $x$. We round each fractional value up to the closest integer, which causes the solution to become an infeasible one, call it $y = \lceil x \rceil$. Then, to retain feasibility, we go through each of the $k$ constraints, pick the $c$-smallest set of $k$ items from $y$ whose deletion causes the constraint to again become satisfied; and we delete the union of these sets from $y$, obtaining $z$. Each set has $c$-cost at most $\frac{k}{\gamma} c(y)$ since for each constraint $i$, any $k$ elements from $g^i$ form an eligible set for deletion, and $y \subset g^i$ consists of at least $\gamma$ items. Thus $c(z) \ge c(y)-k\frac{k}{\gamma} c(y) = (1-k^2/\gamma)c(x)$. Taking $\gamma = k^2/\epsilon$ (compared to the previous $k/\epsilon$), we get the desired result.

\section{Discussion}
We believe that the main result is a nice theoretical illustration of techniques (filtering, rounding, disjunctive programming). However, it remains to be seen if it could be given useful applications. The disjunctive programming trick is definitely senseless sometimes: if we want to write an LP-based computer program to $(1+\epsilon)$-approximately solve multidimensional knapsack instances, it is more efficient to consider the LP corresponding to each guess separately (as in \cite{FC84}) rather than solve the gigantic LP obtained by merging them together. Sometimes an LP-relative~\cite{KPP08} (or Lagrangian-preserving~\cite{CRW04,ABHK09}) approximation algorithm can be used as a subroutine in ways that a non-LP-relative one could not. However, at least in \cite{CRW04,KPP08,ABHK09}, the analysis relied on LP-relative or Lagrangian-preserving analysis of the na\"ive LP, and an arbitrary LP would not have fared as well, and the LP we build here seems not to be useful in this way.

Finding a compact formulation for $k$-dimensional covering knapsack with small integrality gap and such that the LP does \emph{not} depend on the objective function is an interesting open problem. For example, we are not aware of any polynomial-sized extended LP for 1-dimensional covering knapsack with constant integrality gap, in sharp contrast to the packing case. A partial result for $k$-dimensional covering knapsack is the knapsack-cover LP~\cite{CFLP00} (see also~\cite{PC10,CGK10,CCKN10} for applications); it has integrality gap at most $2k$, and while it is not polynomial size, it can be $(1+\epsilon)$-approximately separated~\cite{CFLP00} and hence $(1+\epsilon)$-approximately optimized~\cite{GK07,GLS88} in polynomial time.

From a theoretical perspective, it also seems challenging to find an LP for 2-dimensional (packing) knapsack where the size of the LP is a function of $1/\epsilon$ times a polynomial in $n$, as was done in \cite{BM08} for the 1-dimensional version.

\subsection*{Acknowledgments}
We thank Laura Sanit\`a for helpful discussions on this topic.

\bibliography{../../../huge}
\bibliographystyle{abbrv}

\typeout{Label(s) may have changed. Rerun}
\end{document}